\documentclass[extras,11pt]{article} 
\usepackage{fullpage}
\usepackage{wrapfig}
\usepackage[T1]{fontenc}
\usepackage{ae,aecompl}
\usepackage{amsmath,amssymb}
\usepackage{amsthm}
\usepackage{hyperref}
\usepackage{calc}
\usepackage{tikz}
\usepackage{bold-extra}
\usepackage{subfigure}
\newtheoremstyle{mytheorem}{\topsep}{\topsep}{\sffamily}{}{\bfseries}{.}{.5em}{}
\theoremstyle{mytheorem}
\newtheorem{theorem}{Theorem}
\newtheorem{claim}[theorem]{Claim}
\newtheorem{lemma}[theorem]{Lemma}

\renewenvironment{proof}{\par\noindent{\bf Proof.}\hspace{0.5em}}
    {\hfill\qed\vspace{1ex}}

\newenvironment{proofof}[1]{\par\noindent{\bf #1.}\hspace{0.5em}}
    {\hfill\qed\vspace{1ex}}
\DeclareMathOperator{\poly}{poly}
\DeclareMathOperator{\len}{length}
\DeclareMathOperator{\lev}{level}
\DeclareMathOperator{\comp}{comp}
\newcommand{\eps}{\epsilon}
\newcommand{\prob}[1]{\textit{#1}}
\newcommand{\I}{\mathcal{I}}
\newcommand{\DD}{\mathcal{D}}
\newcommand{\K}{\mathcal{K}}

\newcommand{\C}{\mathcal{C}}
\newcommand{\B}{\mathcal{B}}
\newcommand{\G}{\mathcal{G}}
\newcommand{\eL}{\mathcal{L}}
\newcommand{\pr}{\mathbf{Pr}}
\newcommand{\algo}[1]{\textsc{#1}}
\newcommand{\QQ}{{\mbox{\rm\bf Q}}}

\newcommand{\OPT}{\mbox{\rm OPT}}

\makeatletter\newenvironment{algorithm}{%
   \begin{lrbox}{\@tempboxa}\begin{minipage}{\columnwidth-1cm}}{\end{minipage}\end{lrbox}%
   \centerline{\fbox{\usebox{\@tempboxa}}}
}\makeatother
\newcommand{\algtitle}[1]{\vspace{.1cm}\textbf{#1}}
\newcommand{\alginput}[1]{\\\textbf{Input:} \textit{#1}}
\newcommand{\algoutput}[1]{\\\textbf{Output:} \textit{#1}}
{\makeatletter
 \gdef\xxxmark{%
   \expandafter\ifx\csname @mpargs\endcsname\relax 
     \expandafter\ifx\csname @captype\endcsname\relax 
       \marginpar{xxx}
     \else
       xxx 
     \fi
   \else
     xxx 
   \fi}
 \gdef\xxx{\@ifnextchar[\xxx@lab\xxx@nolab}
 \long\gdef\xxx@lab[#1]#2{{\bf [\xxxmark #2 ---{\sc #1}]}}
 \long\gdef\xxx@nolab#1{{\bf [\xxxmark #1]}}
}
\begin{document}

\date{}
\title{Euclidean Prize-collecting Steiner Forest\footnote{A short version of this paper appears in Proceedings of LATIN 2010~\cite{BH10:latin}.}}

\author{MohammadHossein Bateni\thanks{Department of Computer Science, Princeton University, Princeton, NJ 08540; Email: \textsf{mbateni@cs.princeton.edu}.
The author was supported by 
a Gordon Wu fellowship as well as
NSF ITR grants
                      CCF-0205594, CCF-0426582 and NSF CCF 0832797,
                      NSF CAREER award CCF-0237113,
                      MSPA-MCS award 0528414,
                      NSF expeditions award 0832797.}
 \and MohammadTaghi Hajiaghayi\thanks{AT\&T Labs---Research, Florham Park, NJ 07932; Email: \textsf{hajiagha@research.att.com}.}}
\maketitle

\begin{abstract}
In this paper, we consider {\em Steiner forest} and its
generalizations, {\em prize-collecting Steiner forest} and {\em
$k$-Steiner forest}, when the vertices of the input graph are points
in the Euclidean plane and the lengths are Euclidean distances.
First, we present a simpler analysis of the polynomial-time
approximation scheme (PTAS) of Borradaile et
al.~\cite{BKM08:euc-for} for the {\em Euclidean Steiner
forest} problem.
This is done by proving a new structural property and
modifying the dynamic programming by adding a new piece of
information to each dynamic programming state. 
Next we develop a
PTAS for a well-motivated case, i.e., the multiplicative case, of
prize-collecting and budgeted Steiner forest.
The ideas used in the algorithm may have applications in design of a broad class of bicriteria PTASs.
At the end, we
demonstrate why PTASs for these problems can be hard in the general
 Euclidean case (and thus for PTASs we cannot go beyond the multiplicative case).
\end{abstract}
\section{Introduction}

Prize-collecting Steiner problems are well-known
 network design problems with several applications in
expanding telecommunications networks (see e.g.~\cite{JMP00,SCRS}),
cost sharing, and Lagrangian relaxation techniques (see e.g.
\cite{JV01,CRW01}). The most general version of these problems is
called the {\em prize-collecting Steiner forest (PCSF)}
problem\footnote{It is sometimes called
{\em prize-collecting generalized Steiner tree (PCGST)} in the literature.}, in which,
given a graph $G=(V,E)$, a set of (commodity) pairs
$\DD=\{(s_1,t_1),(s_2,t_2), \dots\}$, a non-negative cost function
$c:E\rightarrow \QQ^{\geq 0}$, and finally a non-negative penalty
function $\pi:\DD\rightarrow \QQ^{\geq 0}$, our goal is a
minimum-cost way of buying a set of edges and paying the penalty for
those pairs which are not connected via bought edges. When all
penalties are $\infty$, the problem is the classic APX-hard
\prob{Steiner forest} problem for which the best approximation
factor is $2-\frac{2}{n}$ ($n$ is the number of vertices of the
graph) due to Goemans and Williamson~\cite{GW95}. When all sinks are
identical in the PCSF problem, it is the classic prize-collecting
Steiner tree problem. Bienstock, Goemans, Simchi-Levi, and
Williamson~\cite{BGSW93} first considered this problem (based on a
problem earlier proposed by Balas~\cite{Bal89}) for which they gave
a 3-approximation algorithm.
The current best approximation algorithm for this problem is a
recent 1.992-approximation algorithm of Archer, Bateni, Hajiaghayi,
and Karloff~\cite{ABHK09} improving upon a primal-dual
$\big(2-\frac{1}{n-1}\big)$-approximation algorithm of Goemans
and Williamson~\cite{GW95}. When in addition all penalties are
$\infty$, the problem is the classic \prob{Steiner tree} problem,
which is known to be APX-hard \cite{cr:7} and for which the best
known approximation factor is $1.55$~\cite{cr:31}.

  There are several
$3$-approximation algorithms for the \prob{prize-collecting Steiner forest}
problem using LP rounding, primal-dual, or iterative rounding methods which
are first initiated by Hajiaghayi and Jain~\cite{HJ06}
(see~\cite{BGSW93,Hajiagha:iter}).
Currently the best approximation factor for this problem is a
randomized $2.54$-approximation algorithm \cite{HJ06}. The approach
of Hajiaghayi and Jain has been generalized by Sharma, Swamy, and
Williamson~\cite{SSW07} for network design problems where violating
arbitrary 0-1 connectivity constraints are allowed in exchange for a
very general penalty function.

Lots of attention has been paid to budgeted versions of Steiner problems
as well. In the \prob{$k$-Steiner forest} (or just $k$-forest for
abbreviation), given a graph $G=(V,E)$ and a set of (commodity)
pairs $\DD$, the goal is to find a minimum-cost forest that connects
at least $k$ pairs of $\DD$. The best current approximation factor
for this problem is in
$O(\text{min}\{\sqrt{k},\sqrt{n}\})$~\cite{GHNR07}. On the other
hand, Hajiaghayi and Jain~\cite{HJ06} could transform notorious
\prob{dense $k$-subgraph} to this problem, for which the current best
approximation factor is $O(n^{1/3-\epsilon})$~\cite{FKP01}. The
special case in which we have a root $r$ and $\DD$ consists of all
pairs $(r,v)$ for $v\in V(G)-\{r\}$ is the well-known NP-hard
$k$-MST problem. The first non-trivial approximation algorithm for
the $k$-MST problem was given by Ravi et al.~\cite{ch:18}, who
achieved an approximation ratio of $O(\sqrt k)$. Later this
approximation ratio is improved to a constant by Blum et
al.~\cite{ch:7}. Currently the best approximation factor for this
problem is $2$ due to Garg~\cite{ch:11}.

In this paper, we consider \prob{Euclidean prize-collecting Steiner
forest} and \prob{Euclidean $k$-forest} in which  the vertices of the input graph are
points in the Euclidean plane (or low-dimensional Euclidean
space) and the lengths are Euclidean distances. For the \prob{Euclidean
Steiner tree} problem, Arora~\cite{arora98:ptas} and Mitchell
\cite{cr:26} gave polynomial-time approximation schemes (PTASs).
    Recently Borradaile, Klein and Kenyon-Mathieu~\cite{BKM08:euc-for} claim a PTAS for the
  more general problem of \prob{Euclidean Steiner forest} .

\subsection{Problem definition}

Motivated by the settings in which the demand of each pair is the
product of the weight of the origin vertex and the weight of the
destination vertex in the pair and thus in a sense contributions of
each vertex to all adjacent pairs are the same (e.g., see {\em
product multi-commodity flow} in Leighton and Rao~\cite{LR99} or \cite{Bonsma,KS02}, and
its applications in wireless networks~\cite{MSL08} or
routing~\cite{CKS04,CKS05}), we consider the following multiplicative
version of prize-collecting Steiner forest for the Euclidean case.


In the \prob{Multiplicative prize-collecting Steiner forest
(MPCSF)} problem, given an undirected graph $G(V,E)$ with
non-negative edge lengths $c_e$ for each edge $e\in E$, and also
given weights $\phi(v)$ for each vertex $v\in V$, our goal is to
find a forest $F$ which minimizes the cost
 \[   \sum_{e\in F} c_e  +  \hskip-.5cm\sum_{{u,v\in V}:{\text{ $u$ and $v$ are not connected via $F$\hskip-2.5cm}}} \hskip-.3cm\phi(u)\phi(v).  \]
%
%
Indeed, this is an instance of PCSF in which each ordered vertex pair
$(u,v)$ forms a request with penalty $\phi(u)\phi(v)$.\footnote{We can change the definition to unordered pairs whose treatment requires only a slight modifications of the algorithms.  Currently, each unordered pair $(u,v)$ has a prize of $2\phi(u)\phi(v)$ if $u\neq v$.}
 We may be asked
to \emph{collect a certain prize $S$}, in which case the goal is to
find the forest $F$ of minimum cost for which
\[ \sum_{{u,v\in V}:{\text{ $u$ and $v$ are connected via $F$\hskip-2cm}}} \phi(u)\phi(v) \geq S. \]
Let us call this problem $S$-MPCSF. We show that this is a
generalization of the $k$-MST problem (see Appendix~\ref{sec:kmst}) and thus currently there is no approximation better than 2
for this problem either. When working on the Euclidean case, the
input does not include any Steiner vertices, as all the points of
the plane are potential Steiner points.

A bicriteria $(\alpha,\beta)$-approximate solution for the the
$S$-MPCSF problem is one whose cost is at most $\alpha\OPT$, yet
collects a prize of at least $\beta S$. Our main contribution in
this paper is a bicriteria $(1+\eps,1-\eps')$-approximation
algorithm that runs in time exponential in $1/\eps$ but polynomial
in $n$ and $1/\eps'$. We then use this  algorithm to obtain a PTAS for
MPCSF.

\subsection{Our contribution}\label{sec:contrib}
First of all, we present a simpler analysis for the algorithm of
Borradaile et al.~\cite{BKM08:euc-for} for the
\prob{Euclidean Steiner forest} problem and reprove the following
theorem.
\begin{theorem}\label{thm:sf}
 For any constant $\eps > 0$, there is an algorithm that runs in polynomial time
 and approximates the \prob{Euclidean Steiner forest} problem within $1+\eps$ of the
 optimal solution.
\end{theorem}
This is done by modifying the dynamic programming (DP) algorithm so
that instead of storing paths enclosing the \emph{zones} in the
algorithm by Borradaile et al., we use a bitmap to identify a zone.
The modification results in simplification of the structural
property required for the proof of correctness (See
Section~\ref{sec:struct}). We prove this structural property in
Theorem~\ref{thm:locality}. The proof has some ideas similar to
\cite{BKM08:euc-for}, but we present a simpler charging scheme that
has a universal treatment throughout. Next we give an overview of 
the dynamic programming algorithm in Section~\ref{sec:algo}. 
We have recently come to know that similar simplifications have been independently discovered by
the authors of~\cite{BKM08:euc-for}, too.

Next we extend the algorithm for Euclidean $S$-MPCSF and MPCSF
problems in Section~\ref{sec:multi}.
\begin{theorem}\label{thm:smpcsf}
 For any $\eps, \eps' > 0$, there is a bicriteria $(1+\eps, 1-\eps')$-approximation algorithm
 for the \prob{Euclidean $S$-MPCSF} problem,
  that runs in time polynomial
 in $n, 1/\eps'$ and exponential in $1/\eps$.
\end{theorem}
Notice that $\eps'$ need not be a constant.
In particular, if all weights are polynomially bounded integers,
we can find in polynomial time a $(1+\eps)$-approximate solution
that collects a prize of at least $S$; this can be done by picking
$\eps'$ to be sufficiently small ($\eps'^{-1}$ is still polynomial).
Next we present a PTAS for \prob{Euclidean MPCSF}.
\begin{theorem}\label{thm:mpcsf}
 For any constant $\eps$, there is a $(1+\eps)$-approximation algorithm
 for the \prob{Euclidean MPCSF} problem,
  that runs in  polynomial time.
\end{theorem}


We also study the case of asymmetric prizes for vertices in which
each vertex $v$ has two types of weights (type one and type two) and the
prize for an ordered pair $(u,v)$ is the product of the first type
weight of $u$, i.e., $\phi^s(u)$, and the second type weight of $v$, i.e., $\phi^t(v)$.
This case is especially interesting because it generalizes the multiplicative
prize-collecting problem when we have two disjoint sets $S_1$ and
$S_2$ and we pay the multiplicative penalty only when two vertices,
one in $S_1$ and the other one in $S_2$, are not connected (by
letting for each vertex in $S_1$ the first type weight be its actual
weight and the second type weight be zero and for each vertex in $S_2$
the first type weight be zero and the second type weight be its actual
weight.) 
 After hinting on the arising complications, we show how we
can extend our algorithms for this case as well.

\begin{theorem}\label{thm:asym-smpcsf}
 For any $\eps, \eps' > 0$, there is a bicriteria $(1+\eps, 1-\eps')$-approximation algorithm
 for the \prob{Euclidean Asymmetric $S$-MPCSF} problem,
  that runs in time polynomial
 in $n, 1/\eps'$ and exponential in $1/\eps$. In addition, for any constant $\eps$, there is a $(1+\eps)$-approximation algorithm
 for the \prob{Euclidean Asymmetric MPCSF} problem,
  that runs in  polynomial time.
\end{theorem}

Indeed, the algorithms in Theorem~\ref{thm:asym-smpcsf} can be
extended to the case in which there are a constant number of
different types of weights for each vertex generalizing the case in
which we have a constant number of disjoint sets and we pay the
multiplicative penalty when two vertices from two different sets are
not connected.
Notice that the case of two disjoint sets already generalizes
the \prob{prize-collecting Steiner tree} problem (by considering $S_1=\{r\}$ and $S_2=V-\{r\}$)
whose best approximation
guarantee is currently $1.992$.  

At the end, we present in Section~\ref{sec:challenge} why PCSF and
$k$-forest problems can be APX-hard in the general case (and thus
for PTASs we cannot go beyond the multiplicative case). We conclude
with some open problems in Section~\ref{sec:conclusion}. All the
omitted proofs appear in the appendix.

\subsection{Our techniques for the prize-collecting version}
Here, we summarize our techniques for the multiplicative prize collecting Steiner forest algorithms; see Section~\ref{sec:multi}.
In all those algorithms, we store in each DP state extra parameters,
including 
the sum of the weights, as well as the multiplicative prize already collected in each component.
These parameters enable us to carry out the DP update procedure.
Interestingly, the sum and collected prize parameters have their own precision units.

In the asymmetric version, a major issue is that no fixed unit is good for all sum parameters.  
Some may be small, yet have significant effect when multiplied by others.
To remedy this, we use variable units, reminiscent of the floating-point  storage formats (mantissa and exponent).
To the best of our knowledge, Bateni and Hajiaghayi~\cite{BH09:facility} were the first to take advantage of this idea in the context of (polynomial time) approximation schemes.
The basic idea is that a certain parameter in the description of DP states has a large (not polynomial) range,
however, as the value grows, we can afford to sacrifice more on the precision.
Thus, we store two (polynomial) integer numbers, say $(i, x)$, where $i$ denotes a variable unit, and $x$ is the coefficient: the actual number is then recovered by $x\cdot u_i$.
The conversion between these representations is not lossless, but the aggregate error can be bounded satisfactorily.

In Section~\ref{sec:prize-mpcsf} we consider the problem where the objective is a linear function of penalties paid and the cost of the forest built.
The challenging case is when the cost of the optimal forest is very small compared to the penalties paid.  
In this case, we identify a set of vertices with large penalties and argue they have to be connected in the optimal solution.  
Then, with a novel trick we show how to ignore them in the beginning, and take them into account only after the DP is carried out.

\section{Preliminaries}\label{sec:prelim}
Let $n=|V|$ be the total number of terminals and let 
$\OPT$ be the total length of the optimal solution.
%
%
A \emph{bitmap} is a matrix with 0-1 entries.
Two bitmaps of the same dimensions are called \emph{disjoint} if and only if they do not have value one at the same entry.
%
Consider two partitions $\mathcal{P}=\{P_1, P_2, \dots,P_{|\mathcal{P}|}\}$  and $\mathcal{P'}=\{P'_1, P'_2, \dots,P'_{|\mathcal{P'}|}\}$ over the same ground set.
Then, $\mathcal{P}$ is said to be a \emph{refinement} of $\mathcal{P}'$
if and only if
any set of $\mathcal{P}$ is a subset of a set in $\mathcal{P}'$, namely $\forall P\in\mathcal{P}, \exists P'\in\mathcal{P}': P\subseteq P'$.
%


\begin{figure}[t]
%
\centerline
{
\subfigure[\label{fig:dissection:1}]{
\begin{tikzpicture}
\draw[step=.5cm,gray,ultra thin] (0,0) grid (4,4);
\draw[step=1cm,gray,thick] (0,0) grid (4,4);
\draw[step=2cm,gray!30!black,very thick] (0,0) grid (4,4);
\draw[step=4cm,black,ultra thick] (0,0) grid (4,4);
\foreach \x in {0,1,2,3,4,5,6,7,8}
{
 \draw[fill=blue] (\x*.25,2) circle (.5mm);
 \draw[fill=blue] (\x*.25,0) circle (.5mm);
 \draw[fill=blue] (0,\x*.25) circle (.5mm);
 \draw[fill=blue] (2,\x*.25) circle (.5mm);
}
\end{tikzpicture}
}\hskip 1cm
\subfigure[\label{fig:dissection:2}]{
\begin{tikzpicture}
\draw[step=.25cm,red,ultra thin] (2,2) grid (3,3);
\draw[step=.5cm,gray,ultra thin] (0,0) grid (4,4);
\draw[step=1cm,gray,thick] (0,0) grid (4,4);
\draw[step=2cm,gray!30!black,very thick] (0,0) grid (4,4);
\draw[step=4cm,black,ultra thick] (0,0) grid (4,4);
\end{tikzpicture}
}}
\caption{\subref{fig:dissection:1} An example of a dissection square with depth $3$, and depiction of portals for a sample dissection square with $m=8$;
\subref{fig:dissection:2} the $\gamma\times\gamma$ grid of cells inside a sample dissection square with $\gamma=4$.\label{fig:dissection}}
\end{figure}
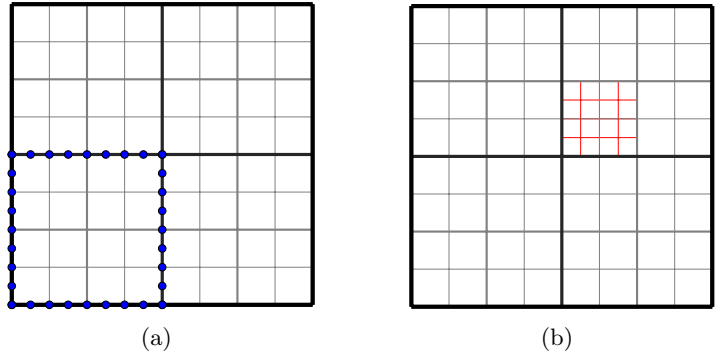
%
By standard perturbation and scaling techniques, we can assume the following conditions hold 
incurring a
cost increase of $O(\eps\OPT)$; see \cite{arora98:ptas,BKM08:euc-for} for example.
\begin{itemize}
\item[(I)] The diameter of the set $V$ is at most $d'=n^2\eps^{-1}\OPT$.  
\item[(II)] All the vertices of $V$ and the Steiner points have coordinates $(2i+1,2j+1)$ where $i$ and $j$ are integers.
\end{itemize}
\vspace{-1.5mm}

For simplicity of exposition, we ignore the above increase in cost.
As we are going to  obtain a PTAS, this increase will be absorbed in the future cost increases.
We have a grid consisting of vertical and horizontal lines with equations $x=2i$ and $y=2j$ where $i$ and $j$ are integers.
Let $\eL$ denote the set of lines in the grid. 
We let $L$ be the smallest power of two greater than or equal to $2d'$
and perform a dissection on the randomly shifted bounding box of size $L\times L$; see Figure~\ref{fig:dissection:1}.

For each dissection square $R$ and each side $S$ of $R$,
designate $m+1$ equally spaced points along $S$ (including the corners) as \emph{portals} of $R$ where $m$ is the smallest power of $2$ greater than $4\epsilon^{-1}\log L$.
So the square $R$ has $4m$ portals.

There is a notion of \emph{level} associated with each dissection square, line, or side of a square.
The bounding box has level zero, and level of each other dissection square is one more than the level of its parent dissection square.
The level of a line $\ell$ is the minimum level of a square $R$ a side of which falls on the line $\ell$.
Thus, the first two lines dividing the bounding box have level one.
If a side $S$ of a square $R$ falls on a line $\ell$, we define $\lev(S)=\lev(\ell)$.
So $\lev(S) \leq \lev(R)$.  
The thickness of the lines in Figure~\ref{fig:dissection} denotes their level:
the thicker the line, the lower is its level.

For a (possibly infinite) set of geometric points $X$, let $\comp(X)$ denote the number of connected components of $X$;
we will use the shorthand ``component'' in this paper.
With slight abuse of notation, $\ell\in\eL$ is used to refer to the set of points\footnote{not necessarily terminals} on $\ell$.
In addition, we use $\eL$ to denote the union of points on the lines in $\eL$.
Similarly, we use $R$ to denote the set of all points on or inside the square $R$.
The set of points on (the boundary of) the square $R$ is referred to by $\partial R$.
The total length of all line segments in $F$ is denoted by $\len(F)$.


The following theorem is mentioned in \cite{BKM08:euc-for} in a stronger form.
We only need its first half  whose proof follows from \cite{arora98:ptas}.
\begin{theorem}{\cite{BKM08:euc-for}}
\label{thm:twoprops}
There is a solution $F$ having expected length at most $(1+\frac{1}{4}\epsilon)\OPT$ such that each dissection square $R$ satisfies the following two properties:
for each side $S$ of $R$, $F\cap S$ has at most $\rho=O(\eps^{-1})$ non-corner components\footnote{Non-corner components are those not including any corners of squares.  Note that each square can have at most four corner components.} (\emph{boundary components property}); and
each component of $F\cap\partial R$ contains a portal of $R$ (\emph{portal property}).
\end{theorem}

\section{Structural theorem}\label{sec:struct}
Let $R$ be a dissection square.
Divide $R$ into a regular $\gamma\times\gamma$ grid of \emph{cells}, where $\gamma$ is a constant power of two determined later; see Figure~\ref{fig:dissection:2}.
We say $R$ is the \emph{owner} of these cells.
The level of these cells, as well as the new lines they introduce, is defined in accordance with the dissection.
That is, we assign them levels as if they are normal dissection squares and we have continued the dissection procedure for $\log\gamma$ more levels.
There are several lemmas in the work of \cite{BKM08:euc-for} to prove the structural property they require (this is the main contribution of that work).
We modify the dynamic programming definition such that its proof of correctness needs a simpler structural property.
The proof of this property is simpler than that in the aforementioned paper.
\begin{theorem}\label{thm:locality}
There is a solution $F$ having expected length at most $(1+\frac{1}{2}\epsilon)\OPT$ such that each dissection square $R$ satisfies the \emph{locality property}:
if the terminals $t_1$ and $t_2$ are inside a cell $C$ of $R$ and are connected to $\partial R$ via $F$,
 then they are connected in $F\cap R$.
\end{theorem}
The proof has ideas similar to \cite[Theorem 3.2, and Lemmas 3.3, 3.4, 3.5 and 3.9]{BKM08:euc-for}.  
We first mention and prove a lemma we need in order to prove Theorem~\ref{thm:locality}.
The lemma more or less appears in \cite{arora98:ptas,BKM08:euc-for}.
\begin{lemma}\label{lem:tot-int}
For the forest $F$ output by Theorem~\ref{thm:twoprops},
$\comp(F\cap\eL)\leq \len(F)$.
\end{lemma}

We can now prove the main structural result.
A side $S$ of a square $R$ is called \emph{private} if it does not lie on a side of the parent square $R'$ of $R$.
Observe that out of any two opposite sides of a dissection square, exactly one is private.

\begin{proofof}{\proofname\ of Theorem~\ref{thm:locality}}
We start with a solution $F$ satisfying Theorem~\ref{thm:twoprops}.
The final solution is produced by iteratively finding the smallest cell $C$ owned  by a square $R$ that violates the locality property,
and adding $\sigma(C,F)$ to $F$, where $\sigma(C,F)$ 
is defined as the union of 
the private sides of $C$
and
any side of $C$ having non-empty intersection with $F$.
We claim the locality property is realized after finitely many such additions.
If after adding $\sigma(C,F)$ to $F$, the cell $C$ still violates the locality property, there has to  be exactly two opposite sides of the cell having non-empty intersection with $F$; otherwise, the $\sigma(C,F)$ is clearly connected.
However, in case of the opposite sides, one middle side will be a private side of $C$ and hence included as well.

Next, we argue that the conditions of Theorem~\ref{thm:twoprops} still hold.
Take a side $S$ of any square $R$.
If the conditions are to be affected for $S$,
it has to be due to an addition involving a cell $C$ that has a side $S'$ such that
(1) $S'$ has non-empty intersection with $S$, and (2) $S'$ is added to $F$ as part of $\sigma(C,F)$.
The condition will be trivial if $S'$ contains $S$.
Thus, we assume that $C$ is a smaller square than $R$.
So $S'$ cannot be a private side of $C$.
However, the number of components on $S$ cannot increase if $S'$ has already an intersection with $F$.

Finally we  show that the additional length is not large.
Let $F^\ast=F\cap \eL$,
and let $\G=\{(x,y): x=2i, y=2j\}$ be the set of all grid points.
We will charge the additions to the connected components of $F^\ast-\G$.
Notice that
\begin{align}
               \comp(F^\ast-\G)
                                 &\leq \comp(F^\ast) + 3|F^\ast\cap\G| \label{eqn:2}\\
                                 &\leq \comp(F^\ast) + 3\cdot(\len(F^\ast) + \comp(F^\ast)) \label{eqn:3}\\
                                 &= 4\comp(F^\ast) + 3\len(F) \nonumber\\
                                 &\leq 7\len(F), &\text{by Lemma~\ref{lem:tot-int}}. \label{eqn:5}
\end{align}
Inequality~\eqref{eqn:2} holds because removal of each grid point on $F^\ast$ increases the number of components by at most three.
To obtain \eqref{eqn:3}, notice that in any connected component of $F^\ast$,
the distance between any two points of $F^\ast\cap\G$ is at least $2$.
Hence, if there are more than one such points, there cannot be more than $\len(F^\ast)$ ones.

We charge this addition to a connected component of $(\partial R\cap  F)-\G$,
 in such a way that each connected component is charged to at most twice: once from each side.
For simplicity, we duplicate each connected component of $(F\cap\ell)-\G$: they correspond to squares from either side of $\ell$.
For any dissection square $R$, let  $\mathcal{C}_{R}$ refer to the connected components of $F\cap R$ that reach $\partial R$.
Further, let $\mathcal{K}_R$ be the set of connected components of $(F\cap\partial R)-\G$.
When $\sigma(C,F)$ is  added where $R$ is the owner of $C$,
there are $k\geq 2$ components $c_1, \dots, c_k\in\C_R$ that become connected.
Any element of $\K_R$ connected via $F\cap R$ to a component $c\in\C_R$
is said to be an \emph{interface} of $c$.
The addition will be charged to a \emph{free} interface of some $c\in\C_R$
with maximum level. This element will no longer be free
for the rest of the procedure.
We argue this procedure  successfully charges all the additions to appropriate border components.
To this end, we shortly prove the following stronger claim via induction on the number of additions performed.
We call a dissection square $R$ \emph{violated} if the locality property does not hold for a cell $C$ owned by $R$.
\begin{claim}\label{clm:cij}
At all times during the execution of this procedure,
any component $c\in\C_R$ has a free interface,
for each violated square $R$.
As a result, any addition can be charged to a free component.
\end{claim}
%
The second statement of the claim follows from the first part.
The first part is proved as follows.
The claim clearly holds at the beginning, 
since all interfaces are free, and each component has an interface.
Suppose the addition $\sigma(C,F)$ is performed and let $R$ be the owner of $C$.
We show any dissection square $R'$ will stay fine.
Notice that the size of the squares $R$ for which the addition is performed is increasing in time.
Hence, any dissection square $R'$ smaller than $R$ is irrelevant in the statement of the claim, since they cannot be violated.
For $R$ itself, each $c_i$ has at least one free interface.
One of the interfaces is used, and thus the new component formed by their union has a free interface.
Suppose for the sake of reaching a contradiction that
a component $c'\in\C_{R'}$ has no free interface after the addition.
Thus $R'$  contains $R$,
and the charging was not done to a private side of $R$.
Recall that prior to the addition, $c'$ is connected to some components of $\C_R$ with at least two free interfaces in $R$.
One of them still remains free.
We charged to the interface of maximum level and it was in $\partial R'$.
Hence, the free interface is also in $\partial R$, leading to a contradiction.

Let (the random variable) $c_{\ell,j}$ denote the number of charges to components on $\ell\in\eL$ due to cells $C$ owned by squares $R$ of level $j$.
Independently of the randomness $\sum_\ell\sum_j c_{\ell,j}\leq 2\comp(F^\ast-\G)$ by the above discussion and Claim~\ref{clm:cij}.
Note the cost of adding $\sigma(C,F)$ (charged to a component on $R$)
is at most $4L'/\gamma$ where $L'$ is the side length of $R$.
The total increase due to charges to $\ell$ is at most $\sum_{j\geq\text{depth}(\ell)}c_{\ell,j}\frac{4L}{\gamma 2^j}$ where $L$ is the side length of  the bounding box. 
Due to the randomization in the dissection, 
we have
$\pr[\text{depth}(\ell)=i]=2^i/L$; see \cite{arora98:ptas} for instance.
The expected increase in length is thus
\begin{align*}
\sum_\ell\sum_i\frac{2^i}{L}\sum_{j\geq i}c_{\ell,j}\frac{4L}{\gamma 2^j}
               &\leq\frac{4}{\gamma}\sum_\ell\sum_j\frac{c_{\ell,j}}{2^j}\sum_{i\leq j}2^i   \\
               &
\leq\frac{8}{\gamma}\sum_\ell\sum_jc_{\ell,j}  \\
              &
\leq\frac{16}{\gamma}\comp(F^\ast-\G) &\text{by Claim~\ref{clm:cij}}\\
               &\leq\frac{112}{\gamma}\len(F) &\text{by \eqref{eqn:5}}.
\end{align*}
\noindent We pick $\gamma$ to be the smallest power of two larger than $112(1+\epsilon)\cdot 2\epsilon^{-1}$
to finish the proof.
\end{proofof}

Therefore, with probability $1/2$, we have $\len(F)\leq(1+\eps)\OPT$.
In the entire argument, no attempt was made to optimize the parameters.


\section{The algorithm}\label{sec:algo}
A \emph{subsolution} for $R$ is a finite set of line segments $F \subset R$
satisfying conditions of Theorems \ref{thm:twoprops} and \ref{thm:locality},
with the extra property that any terminal $t$ in $R$ is connected via $F$ either to its mate or to $\partial R$.
A \emph{configuration} $\chi=(\mathcal{K}, \mathcal{P})$ for $R$ has two portions:
a set $\mathcal{K}$ of pairs $\kappa_i=(P_i,M_i)$ and a partition $\mathcal{P}$ whose ground set is $\mathcal{K}$, such that
\begin{itemize}
\item $P_i$ is a subset of portals of $R$;
\item $M_i$ is a bitmap of size $\gamma\times\gamma$;
\item $P_i$ and $P_j$ are disjoint if $i\neq j$;
\item the total number of portals, namely $\sum_i|P_i|$, is at most $4(\rho+1)$; and
\item bitmaps $M_i$ and $M_j$ are disjoint if $i\neq j$.
\end{itemize}
The configuration captures sufficient information about $F$ so as to make it possible to take care of the interaction between $R$ and the outside.
In particular, each pair $(P,M)$ describes a connected component of $F$, by specifying the set of portals on its boundary and the set of cells connected to these portals.
Roughly speaking, the partition $\mathcal{P}$ tells us which components $\kappa_i$ and $\kappa_j$ need to be connected from outside $R$:
 this implies the existence of a pair of terminals that are in $\kappa_i$ and $\kappa_j$, respectively, but they are not connected in $R$.
We will see below why this restrictive abstraction does not lose any crucial subsolutions.
%

We say a subsolution $F$ is \emph{compatible} with a  configuration $\chi=(\mathcal{K},\mathcal{P})$ if
\begin{enumerate}
\item
for any connected component $\kappa$ of $F$ that intersects $\partial R$,
 there exists a pair $\kappa'=(P,M)\in\mathcal{K}$ such that
\begin{itemize}
\item $\kappa$ spans $P$;
\item each connected component of $\kappa\cap\partial R$ contains a portal of $P$;
\item the bitmap $M$ has value one in the positions corresponding to any cell $C$ containing a terminal $t$ of $\kappa$; and
\end{itemize}
\vspace{-1mm}
\item
 any terminal pair located in different components $\kappa_1$ and $\kappa_2$ of $\mathcal{K}$ are either connected via $F\cap R$,
 or $\kappa_1$ and $\kappa_2$ are in the same set of $\mathcal{P}$.
\end{enumerate}

\subsection{The dynamic programming}

\begin{figure}
\begin{algorithm}
\algtitle{Algorithm \algo{EuclideanSteinerForest}}
\alginput{Set of terminals $V$ in the plane, and set $\DD$ of pairs of terminals}
\algoutput{A forest $F$ connecting pairs in $\DD$}
\begin{enumerate}\setlength{\itemsep}{-.00in}
\item Carry out the perturbation and scaling.
\item Let $L$ be smallest power of two larger than $2n^2\eps^{-1}d$, where $d$ is the maximum distance of a pair.
\item Perform a random dissection in the bounding box of side $L$.
\item Place $m+1$ portals on each side of a dissection square, where $m$ is the smallest power of two larger than  $4\eps^{-1}\log L$.
\item Solve the base cases $T_R[\chi]$ for leaf dissection squares $R$:\\
Go over all possible ways of connecting the portals and the center point.
\item Populate the table $T_R[\chi]$ in increasing order of size for $R$:\\
      For any $\chi=(\mathcal{K},\mathcal{P})$ corresponding to $R$ consisting of $R_1,\dots,R_4$:
\begin{enumerate}\setlength{\itemsep}{-.02in}
 \item Go over all configurations $\chi_i=(\mathcal{K}_i,\mathcal{P}_i)$ corresponding to $R_i$.
 \item Build $\mathcal{K}^1$ from the union of all components of $\mathcal{K}_i$ with expanded bitmaps.
 \item Build $\mathcal{P}'$ from the union of $\mathcal{P}_i$.
 \item If there is a terminal pair $(t_1,t_2)$ where $t_1\in R_{i_1}$ and $t_2\in R_{i_2}$ for $i_1\neq i_2$,
\begin{itemize}\setlength{\itemsep}{-.02in}
 \item If there is no bitmap in $\chi_{i_1}$ (or $\chi_{i_2}$) containing the cell containing $t_1$ (or $t_2$ respectively), the configuration is bad.
 \item Otherwise, merge the sets corresponding to the appropriate components in $\mathcal{P}'$.
\end{itemize}
 \item Build $\mathcal{K}^2$ by merging components having the same portals,
  and make appropriate changes to $\mathcal{P}'$.
 \item Build $\mathcal{K}^3$ by removing portals not on $\partial R$.
 \item If any component with empty portal set has unsatisfied connectivity requirement in $\mathcal{P}'$, the current configurations are not consistent.
 \item Build $\mathcal{K}^4$ by eliminating components with empty portal set.
 \item If any bitmap contradicts the locality property, these configurations are not consistent.
 \item If the configurations are consistent, update $T_R[\chi]$ with .
\[ \min \left\{ T_R[\chi] + \sum_{i=1}^4 T_{R_i}[\chi_i]\right\}.\]
\end{enumerate}
\item Find the final solution among $T_R[\chi]$ where $R$ is the bounding box and $\chi$ has no unsatisfied requirement.
\item Construct the solution $F$ by recursively following the values from $T_R[\chi]$.
\end{enumerate}
\end{algorithm}
\caption{The algorithm for \prob{Euclidean Steiner forest} problem.\label{fig:for-alg}}
\end{figure}

In the dynamic program, we build a table $T_R[\chi]$, indexed by configurations for each dissection square $R$.
The goal is to populate this table so that $T_R[\chi]$ is the minimum length of a subsolution for $R$ that is compatible with $\chi$.
First of all, we show that for each $R$, the number of configurations is small.
Consider $\chi=(\mathcal{K},\mathcal{P})$.
There are at most $\lambda=4(\rho+1)$ pairs in $\mathcal{K}$.
For a particular $\kappa=(P,M)$, there are
$\sum_{i=0}^{\lambda}{m+1\choose i}=O(m^{\lambda+1})$
 options for the set of portals $P$.
The bitmap $M$ has $2^{\gamma^2}$ possibilities.
A crude upper bound of $2^{\lambda^2}$ is trivial for possibilities of $\mathcal{P}$.
Thus, the total number will be at most
\[  \Phi = \left[ O\left(m^{\lambda+1}\right) \cdot 2^{\gamma^2} \right]^{\lambda} \cdot 2^{\lambda^2} = O(\poly(m))=O(\poly\log(n)). \]
Theorems \ref{thm:twoprops} and \ref{thm:locality} guarantee the existence of a near-optimal solution all whose subproblems are compatible with a  configuration:
The connected components of $F$ reaching $\partial R$ can be decomposed into disjoint bitmaps
because of Theorem~\ref{thm:locality}.
Theorem~\ref{thm:twoprops} on the other hand ensures each connected component on $\partial R$ contains a portal, and the total number of such components is small.
The details of the DP update, as well as its correctness proof, appears 
below.  



The final solution of the problem is obtained from the minimum $T_R[\chi]$ where $R$ is the bounding box, and $\mathcal{P}$ of $\chi$ does not require any connections: i.e., all sets of the partition are singletons.
This would imply all the necessary connections have been made inside $R$.
To actually construct the solution, we need to store additional information in each dynamic programming state indicating which configurations it was last updated from.
It is then straightforward to recursively construct the solution, by taking the union of the pertinent configurations.


Here we show how the dynamic programming table for Euclidean prize-collecting Steiner forest 
is updated from the already-computed values.
And finally we show why the update routine is sound and complete.
The table $T_R[\chi]$ is populated in the order of increasing size for $R$.
For a base dissection square $R$, finding the value of $T_R[\chi]$ is straightforward.
Notice that there is at most one point (possibly with several terminals collocated) inside $R$.
Depending on whether the mates of those terminals are collocated with them or not, we may need to connect some of them to the boundary $\partial R$.
There are only a constant number of portals in $\chi$, hence we can go over all the ways to connect them up and find the smallest value.
Note that there cannot be any Steiner point inside $R$.

Now we get to the update rule.
Consider a dissection square $R$ and a corresponding configuration $\chi=(\mathcal{K},\mathcal{P})$.
Let $R_i$ for $i=1,2,3,4$, be the children of $R$ in the dissection.
Take corresponding configurations $\chi_i=(\mathcal{K}_i,\mathcal{P}_i)$.
Notice that each cell of $R$ consists of exactly four cells of one $R_i$.
We can expand a bitmap $M$ of $R_i$ to a bitmap $M'$ of dimensions $2\gamma\times 2\gamma$ for $R$,
by placing three all-zero bitmaps of dimensions $\gamma\times\gamma$ at appropriate locations around $M$.
We do this in such a way that the portion corresponding to $M$ still points to $R_i$ inside $R$.
Consider all the components $\kappa=(P,M)$ corresponding to the four subsquares, expand their bitmap, and collect them in $\mathcal{K}^1$.
Merge the partitions $\mathcal{P}_i$ to get $\mathcal{P}'$.
If there is a terminal pair $(s,t)$ where $s$ is in $R_i$ and $t$ is in a different $R_j$,
there should be a component corresponding to each of these in $R_i$ and $R_j$, respectively.
Otherwise, these configurations do not correspond to any (valid) subsolution.
Merge the sets corresponding to these components in $\mathcal{P}'$: i.e., they have to be connected.
Next merge any two components of $\mathcal{K}^1$ if they share a portal, and build $\mathcal{K}^2$.
Further, make appropriate changes in $\mathcal{P}'$.
Build  $\mathcal{K}^3$ by removing from $\mathcal{K}^2$ all portals not on $\partial R$.
Some of these components reach $\partial R$ and some do not, namely those with an empty portal set $P$.
If there is any component with empty portal set that is not  one partition set,
we deem the configurations $\chi_i$ as \emph{inconsistent}:
in this case, some components that are required to be connected together do not reach the boundary.
Otherwise, remove all the pairs in $\mathcal{K}^3$ with empty portal set to obtain $\mathcal{K}^4$.
Now, if there is a cell of $R$ whose four constituent cells reach the boundary as more than one connected component, the configurations are not consistent either: this contradicts the property of Theorem~\ref{thm:locality}.
Finally, reduce the dimensions of the bitmaps to $\gamma\times\gamma$ such that a cell of the new bitmap acquires value one if and only if there is a one in one of the positions corresponding to the constituent cells in the original bitmap.
Now, $\chi = (\mathcal{K}, \mathcal{P})$ is said to be consistent with the four configurations $\chi_1,\dots,\chi_4$ if and only if $\mathcal{P}$ contains all the requirements of $\mathcal{P}'$, i.e., $\mathcal{P}'$ is a refinement of $\mathcal{P}$, and in addition, there exists a $\kappa=(P,M)\in\mathcal{K}$ for any $\kappa'=(P',M')\in\mathcal{K}'$ such that $P\subseteq P'$ and $M = M'$.  
In case these configurations are consistent,
 $T_R[\chi]$ will take the minimum of its current value and $\sum_i T_{R_i}[\chi_i]$.
You can refer to Figure~\ref{fig:for-alg} for a summary.

\subsection{Proof of correctness}
Correctness follows from induction on the size of the square $R$ that all dynamic programming states have their intended value.
In particular, we know that there is a near-optimal solution all whose subsolutions are compatible with one configuration.
Hence, these will be computed correctly and give the final solution.
More specifically the following claim holds for all DP states.
\begin{lemma}\label{lem:dp-sf}
A dynamic programming state $T_{R}[\chi]$ ends up having the minimum value corresponding to a solution $F$ of $R$, such that for any dissection square $R'$ which is a descendant of $R$ in the dissection tree, the subsolution $F\cap R'$ of $R'$ is compatible with a configuration $\chi'$ for $R'$.
\end{lemma}

Now, we are at the position to prove the main Theorem regarding the \prob{Euclidean Steiner forest} problem.
\begin{proofof}{\proofname\ of Theorem~\ref{thm:sf}}
By Lemma~\ref{lem:dp-sf}, the proposed dynamic programming is sound and complete.
There are $\Phi=O(\poly(n))$ DP states.
To solve each non-base state, we go over at most $\Phi^4$ child states and then perform a polynomial consistency check.
Each base case state is computed in constant time.
Hence, the total algorithm runs in time $O(\poly(n))$.
\end{proofof}


\subsection{Highlights of the new ideas}\label{sec:newideas}
Here, we point out the differences between our work and the previous work of \cite{BKM08:euc-for}.
Borradaile et al.\ use closed paths to identify the connected zones of the dissection square.
These paths consist of vertical and horizontal lines and all the break-points are the corners of the cells.
As part of their structural property, they prove that they can guarantee a solution in which these zones can be identified via paths whose total length is at most a constant $\eta$ times the perimeter of the square $R$.
Then each path is represented by a chain of $\{1,2,3\}$ of length at most $O(\eta\gamma)$: the three values are used to denote moving one unit forward, or turning to the left or right.  This results in a storage of $3^{O(\eta\gamma)}$ which is a constant parameter.
Instead, we use a bitmap of size $\gamma\times\gamma$ to address this issue.
Each zone is represented by a bitmap that has an entry one in the cells of the zone.
The bound that we obtain, $2^{\gamma^2}$, may be slightly worse than the previous work, however,
a simpler structural property, namely the locality property, suffices as the proof of correctness.
Borradaile et al.\ in contrast need a bound on the total length of the zone boundaries, as noted above.

In addition to the simplification made due to this change,
both to the proof and the treatment of the dynamic programming,
we simplify the proof further.
Borradaile et al.\ charge the additions of $\sigma(C,F)$ to three different structures,
and the argument is described and analyzed separately for each.
We manage to perform a universal treatment and charging all the additions to the simplest of the three structures in their work.
But this can be done only  after showing $F^\ast-\G$ has a limited number of components.
The proof is simple yet elegant---a weaker claim is proved in \cite{BKM08:euc-for}, but even the statement of the claim is hard to read.

\section{Multiplicative prizes}\label{sec:multi}

We first tackle the \prob{$S$-multiplicative prize-collecting Steiner forest} problem.
Then, we will take a look at its asymmetric generalization.
Finally, we show how the \prob{multiplicative prize-collecting Steiner forest} problem can be reduced to $S$-MPCSF.

\subsection{Collecting a fixed prize}
\label{sec:smpcsf}

Suppose we are given $S$, the amount of prize we should collect.
Let $\OPT$ be the minimum cost of a forest $F$ that collects a prize of at least $S$,
and suppose $Q\subseteq\DD$ is the set of terminal pairs connected via $F$.
We show how to find a forest with cost at most $(1+\eps)\OPT$  that collects
a prize of at least $(1-\eps')S$.
By the structural property, we know that there is a solution $F'$ connecting the same
 set of terminal pairs $Q$ whose cost is at most $(1+\eps)\OPT$, yet it satisfies the
 conditions of Theorems~\ref{thm:twoprops} and \ref{thm:locality}.
Round all the vertex weights down to the next integer multiple of
$\theta=\eps' \sqrt S/2n$.
In a connected component of $F'$ of total weight $A_i$ that  lost a weight $a_i$ due to
rounding, the lost prize is $A_i^2-(A_i-a_i)^2 \leq 2a_iA_i\leq 2a_i\sqrt S$,
because the total weight of the component is at most $\sqrt S$.
Thus, $F'$ collects at least $S - 2n\theta\sqrt S \leq (1-\eps')S$ from the rounded weights.


Each dynamic programming state consists of a dissection square $R$,
a set of components $\mathcal{K}$, and a new parameter $\Pi$
which denotes the total prize collected inside $R$ by connecting
the terminal pairs.
Each element of $\mathcal{K}$---corresponding to a connected component in the subsolution---now has the form $\kappa=(P, \Sigma)$
where $P$ denotes the portals of $\kappa$, 
and $\Sigma$ is the total sum of the weights in $\kappa$.
The DP is carried out in a fashion similar to that of \cite{arora98:ptas}.
The values of $\Sigma$ and $\Pi$ are easy to determine for the base cases.
It is not difficult to update them, either.
Whenever two components $\kappa_1=(P_1, \Sigma_1)$ and $\kappa_2=(P_2, \Sigma_2)$
merge in the DP, the sum $\Sigma$ for the new component is simply $\Sigma_1+\Sigma_2$.
Besides, the merge increases the $\Pi$ value of the DP state by $2\Sigma_1\Sigma_2$.

\begin{proofof}{\proofname\ of Theorem~\ref{thm:smpcsf}}
The soundness and completeness is simple and is along the same lines as the proof of Theorem~\ref{thm:sf}.
Carrying out the above operation assumes the values of $\Sigma$ and $\Pi$ could be stored
accurately.  However, as they describe the dynamic programming states, their size should be
sufficiently small or else the algorithm will not run in polynomial time.
Here does the rounding help us.
All values of $\Sigma$ are stored as multiples of $\theta$ and the values of $\Pi$ are stored
as multiples of $\theta^2$.  Notice that as we round the vertex weights at the beginning,
throughout the algorithm the values of $\Sigma$ and $\Pi$ will be multiples of their respective
units.  Hence, no extra precision error will occur and we find the aforementioned solution.
If at any time during the execution of the algorithm, the value of $\Sigma$ goes above $\sqrt S$,
we truncate it to $\sqrt S$.  Similarly, the value of $\Pi$ is not allowed to surpass $S$.
This does not eliminate any solution, because at the point of truncation, the subsolution has
already gathered sufficient prize.
Hence, the range of $\Sigma$ is from zero up to $\sqrt S$, and this gives
$\sqrt S/\theta = 2n/\eps'$ different values.
Similarly for $\Pi$, there are at most $S/\theta^2 = 4n^2/\eps'^2$ options.
There are at most
\[  \Phi_1 = \left[ O\left(m^{\lambda+1}\right) \cdot 2^{\gamma^2} \cdot 2n/\eps' \right]^{\lambda} \cdot 2^{\lambda^2} \cdot 4n^2/\eps'^2 = O\left(\poly\Big(n,\frac{1}{\eps'}\Big)\right) \]
\noindent DP states for each square $R$.
The running time is polynomial in $\Phi_1$ and the claim follows.
\end{proofof}

To start the algorithm, we need to guarantee the instance satisfies
the conditions at the beginning of Section~\ref{sec:prelim}.
See Appendix~\ref{sec:prelim-mult} for details of how this is achieved.

\subsection{The asymmetric prizes}

The basic idea is to store two parameters $\Sigma^s$ and $\Sigma^t$ for each component of $\mathcal{K}$.
These parameters store the total weight of the first and second type in the component, namely $\sum_i\phi^s_i$ and $\sum_i\phi^t_i$, respectively.
The difficulty is that to collect a prize of $A=A^sA^t$ in a component,
only one of the parameters $A^s$ or $A^t$ needs to be large.
In particular, we cannot do a rounding with a precision like $\eps'\sqrt A/n$.
It may even happen that  $A^s$ is large in one component, whereas we have a large $A^t$ in another.
In fact, we cannot store the values of the $\Sigma^s$  or $\Sigma^t$ as multiples of a fixed unit.
To get around the problem, $\Sigma^s$ is stored as a pair $(v,x)$, where $v$ is  a vertex of the graph
and $x$ is an integer.
Together they show that $\Sigma^s$ is $x\cdot\eps_1\phi^s(v)/n^2$;
the value of $\eps_1$ will be chosen later, and $v$ is supposed to be the vertex of largest
type-one weight present in the component.
A similar provision is made for $\Sigma^t$.
Finally, the value of $\Pi$ is stored as a multiple of $\eps_2 A/n$;
we will shortly pick the value of $\eps_2$.

Whenever $\Sigma^s_1=(v_1,x_1)$ and $\Sigma^s_2=(v_1,x_1)$ are added to give
$\Sigma^s=(v,x)$,  we do the calculation as follows:
let $v$ be the vertex  $v_1$ or $v_2$ that has the larger $\phi^s$ value,
and then 
$$x=\left\lfloor \frac{x_1\phi^s(v_1)/n^2 +  x_2\phi^s(v_2)/n^2}{\eps_1\phi^s(v)/n^2} \right\rfloor.$$

\begin{proofof}{\proofname\ of Theorem~\ref{thm:asym-smpcsf}}
The precision error for $\Sigma^s=(v,x)$ is at most
$n\cdot\eps_1\phi^s(v)/n^2=\eps_1\phi^s(v)/n$, because there is an accumulation of at most $n$
rounding errors each of which has been less than $\eps_1\phi^s(v)/n^2$.
Notice that if $\Sigma^s$ is stored in terms of the vertex $v$, it has to include $v$ and
thus its type one weight is at least $\phi^s(v)$.  Hence, the precision error is at most a
$\eps_1/n$ multiplicative factor.
Therefore, when we do a multiplication of $\Sigma^s\Sigma^t$ to get an addition to $\Pi$,
the error is at most a multiplicative $2\eps_1/n$:
$(1-\eps_1/n)A^s(1-\eps_1/n)A^t\geq(1-2\eps_1/n)A^sA^t$.
Next a rounding error may happen to store the value in terms of $\eps_2 A/n$.
Each $\Pi$ on the other hand is made up of at most $n$ addition terms,
so the total error is at most $n(2\eps_1/n+\eps_2/n)A$.
We pick $\eps_1=\eps_2=\eps'/3$ to conclude that the total error is bounded by
$\eps' A$.

All the discussion applies to $\Sigma^t$ as well.
Due to truncation and rounding, there are at most $n/\eps_2$ options for $\Pi$.
And each $\Sigma^s$ (or $\Sigma^t$) has at most $n^2/\eps_1$ possibilities.
Thus, the total number of DP states for each dissection square is $\Phi_2=\poly(n,1/\eps')$.
Therefore, we obtain a bicriteria approximation to the asymmetric variant of the problem.
\end{proofof}

\subsection{The prize-collecting version: trade-off between penalty and forest cost}\label{sec:prize-mpcsf}

In the prize-collecting variant, we pay for the cost of the forest,
and for the prizes not collected.
If the total weight is $\Delta$,
the prize not collected is $\Delta^2$ minus the collected prize.
One difficulty here is to determine the correct range for the collected prize
so that we can use the algorithm of Section~\ref{sec:smpcsf}.
The trivial range is zero to $\Delta^2$.
However, the rounding precision we pick for the penalties should also take into
account the cost of the forest.  If the cost of the intended solution is much smaller than
$\Delta^2$, we cannot simply go with rounding errors like $\eps\Delta/n$.
Otherwise, the error caused due to rounding the penalties will be too large compared to the solution value.

The trick is to find an estimate of the solution value, and then consider two cases depending on
how the cost compares to the total penalty.
Using a $3$-approximation  algorithm, we obtain a solution of value $\omega$.
We are guaranteed that $\OPT \geq \omega/3$.
If $\Delta^2 \leq \omega/3$,
 the optimum solution is to collect no prize at all.
Otherwise, assume $\Delta^2 > \omega/3$.
To beat the solution of value $\omega$, we should collect
a prize of at least $\Delta^2-\omega$.

We first consider the simpler case when $\omega/\Delta^2 > 1/n^2$:
For an $\eps' >0$ whose precise value will be fixed below,
we use the algorithm of Section~\ref{sec:smpcsf} to find
a bicriteria $(1+\eps/2,1-\eps')$-approximate solution for collecting
a prize $S$; this is done for any $S$ which is a multiple of $\eps'\Delta^2$
in range $[(1-\eps')\Delta^2-\omega, \Delta^2]$.
We select the best one after adding the uncollected prize to each
of these solutions.
Suppose the optimal solution $\OPT$ collects a prize $S'$.
Let $\OPT_f = \OPT - (\Delta^2 - S')$ be the length of the forest.
Round $S'$ down to the next multiple of $\eps'\Delta^2$, say $S$.
Fed with prize value $S$, the algorithm finds a solution that collects
a prize of at least $(1-\eps')S$ with forest cost at most $(1+\eps/2)\OPT_f$.
\begin{claim}\label{clm:mpcsf:eps}
The total cost of this solution is at most $(1+\eps)\OPT$ if $\eps' = \frac{\min\left( \frac{\eps}{3}, 1  \right)}{6n^2}$.
\end{claim}

\begin{proof} 
The total cost of this solution is
\begin{align}
       \left(1+\frac{\eps}{2}\right)\OPT_f + \left[\Delta^2 - (1-\eps')S\right]
&\leq  \left(1+\frac{\eps}{2}\right)\OPT_f + \left[\Delta^2 - (1-\eps')(1-\eps')S'\right]  \nonumber \\
&\leq  \OPT + \frac{\eps}{2}\OPT_f + (2\eps'+\eps'^2)S'  \nonumber \\
&=  \OPT + \frac{\eps}{2}\OPT + (2\eps'+\eps'^2)\frac{S'}{\OPT}\OPT  \nonumber\\
&\leq  \OPT + \frac{\eps}{2}\OPT + (2\eps'+\eps'^2)\frac{\Delta^2}{\OPT}\OPT   \nonumber\\
&\leq  \OPT + \frac{\eps}{2}\OPT + (2\eps'+\eps'^2){3}{n^2}\OPT  \label{eqn:14}\\
&\leq  \OPT + \frac{\eps}{2}\OPT + \frac{\eps}{2}\OPT  \label{eqn:15}\\
&=     (1+\eps)\OPT,  \nonumber
\end{align}
where \eqref{eqn:14} follows from $\frac{\Delta^2}{\OPT} \leq \frac{n^2\omega}{\omega/3} = 3n^2$,
and \eqref{eqn:15} uses the definition of $\eps'$.
\end{proof}

The other case, i.e., $\omega/\Delta^2\leq 1/n^2$, is more challenging.
Notice that in order to carry out the same procedure in this case, $\eps'$ may not be bounded by $1/\poly(n)$ and thus the running time may not be polynomial.
The solution, however, has to collect almost all the prize.
Thus, one of the connected components includes almost all the vertex weights.
We set aside a subset $\B$ of vertices of large weight.
The vertices of $\B$ have to be connected in the solution, or else the paid penalty will be too large. Then, dynamic programming proceeds by ignoring the effect of these vertices and only keeping tabs on how many vertices from $\B$ exist in each component.
At the end, we only take into account the solutions that gather \emph{all} the vertices of $\B$ in one component and compute the actual cost of those solutions and pick the best one.
In the following, we provide the details of our method and prove its correctness.

Let $\B$ be the set of all vertices whose weight is larger than $n\omega/\Delta$.
\begin{lemma}\label{lem:large-comp}
All the vertices of $\B$ are connected in the optimal solution.
\end{lemma}
\begin{proof}
  There are at most $n$ components, so there is a component, say $\mathcal{C}$,
 whose total weight is not less than $\Delta/n$.
  We claim all the vertices of $\B$ are inside this component.
  The penalty paid by the optimal solution is at most $\omega \leq \Delta^2/n$.
  If there is any vertex of $\B$ outside $\mathcal{C}$, the penalty of
 the solution is more than $\Delta/n\cdot n\omega/\Delta =\omega$,
 yielding a contradiction.
\end{proof}

Next, we round up all the weights to the next multiple of $\theta=\eps' \omega/\Delta$
for vertices not in $\B$.
 Define $\OPT'$ as the optimal solution of the resulting instance.
 Let $\OPT_f$ be the length of the forest in $\OPT$, and define $\OPT'_f$ similarly.
 Let $\OPT_\pi$ and $\OPT'_\pi$ denote the penalty paid by $\OPT$ and $\OPT'$, respectively.
Assume that $\eps' \leq 1$.
\begin{lemma}\label{lem:round-err}
$\OPT'_\pi \leq \OPT_\pi +12n\eps'\OPT$.
\end{lemma}
\begin{proof} 
 We recompute the penalties paid by $\OPT$
 using the rounded weights.
 The pair $(s,t)$ not connected in $\OPT$  is either of the two kinds:
 (1) one of $s$ and $t$ is in $\B$; or (2) none of them is in $\B$.
 The total rounding error for the penalties of the first type is bounded by
 $n\Delta\theta$.
 There are at most $n^2$ pairs of the second type.
  Since the weights of these terminals are at most $n\omega/\Delta$,
 the error is not more than $n^2[2(n\omega/\Delta)\theta + \theta^2]$.
 Hence, the total error is at most
\begin{align*}
       n^2[2(n\omega/\Delta)\theta + \theta^2] +  n\Delta\theta
&\leq   n^2\left[(\eps'^2 + 2n\eps')\frac{\omega^2}{\Delta^2}\right] +  n\eps'\omega \\
&\leq   n^2\left[3n\eps'\frac{\omega^2}{\Delta^2}\right] +  n\eps'\omega & \text{because $\eps' \leq 1$}\\
&=      \left(3n^2\frac{\omega}{\Delta^2} + 1\right)  n\eps'\omega \\
&\leq   4n\eps'\omega  &\text{because $\frac{\omega}{\Delta^2}\leq \frac{1}{n^2}$},
\end{align*}
which is no more than $12n\eps'\OPT$ as desired.
\end{proof}

Suppose we use a dynamic programming approach similar to the previous subsections
to find the approximately minimum forest length for any specified collected prize amount;
in particular, we obtain a bicriteria $(1+\eps/2, 1-\eps')$-approximate solution.
During this process, we ignore the weights associated with vertices in $\B$.
Consider a DP state $\chi=(\K, \Pi)$ corresponding to a dissection square $R$.
Each component $\kappa\in\K$ looks like  $(P, \Sigma, \mu)$:
the new piece of information, $\mu$, is an integer number denoting the number of
vertices of $\B$ inside $\kappa$.
Extending the previous algorithm to populate the new DP table is simple.
Finally, we look at all the configurations $\chi$ for the bounding box such that
the $\mu$ value of one component is exactly $|\B|$ whereas it is zero for all other components.
This guarantees that all elements of $\B$ are inside the former component and hence
we can add up the penalties involving those vertices.
Let $\K=\{\kappa_1, \kappa_2, \dots, \kappa_q\}$ where $\kappa_i=(P_i,\Sigma_i)$,
and let $\kappa_1$ be the component containing $\B$.
The additional cost due to vertices of $\B$ is
$$\left(\sum_{v\in\B}\phi(v)\right)\cdot\left(\sum_{i=2}^{q}\Sigma_i\right).$$
Finally, we report the best solution corresponding to these configurations.

\begin{proofof}{\proofname\ of Theorem~\ref{thm:mpcsf}}
Let us first see that the algorithm described runs in polynomial time.
It is sufficient to bound the number of configurations.
The new piece of information has at most $n$ possibilities.
Further, $\Sigma \leq \frac{n^2}{\eps'}\theta$ is always a multiple of $\theta$.
Similarly, $\Pi$ will not exceed $\frac{n^4}{\eps'^2}\theta^2$ and is always a multiple of $\theta^2$.

 We pick $\eps'=\frac{1}{24n}$.
  By Lemmas~\ref{lem:large-comp} and \ref{lem:round-err}, the rounding does not increase the penalties paid by the optimal solution by more than $\eps/2\OPT$.
 We then utilize the algorithm described for $S$-MPCSF to find a solution of cost at most $(1+\eps/2)\OPT_f + \OPT_\pi +\eps/2\OPT \leq (1+\eps)\OPT$.
Finally,  changing the weights back to the original values clearly does not increase the cost.
\end{proofof}

\section{Evidence for Hardness}\label{sec:challenge}

So far PTASs for geometric problems in Euclidean plane including
ours and those of Arora \cite{arora98:ptas} and Mitchell
\cite{cr:26} can be easily generalized for Euclidean $d$-dimensional
space, for any constant $d> 2$. However we can prove the following
theorem on the hardness of the problem for Euclidean $d$-dimensional
space.

\begin{theorem}
If notorious densest $k$-subgraph is hard to approximate within a
factor $O(n^{\frac{1}{d}})$ for some constant $d$, then for any $d'>
2d+1$, the $k$-forest problem in Euclidean $d'$-dimensional space is
hard to approximate within a factor
$O(n^{\frac{1}{2d}-\frac{1}{d'-1}})$.
\end{theorem}
\begin{proof}
Hajiaghayi and Jain~\cite{HJ06} show that if densest $k$-subgraph is
hard to approximate within a factor $O(n^{\frac{1}{d}})$, then the
$k$-forest problem on stars is hard to approximate within a factor
$O(n^{\frac{1}{2d}})$. On the other hand, Gupta~\cite{Gupta00} shows
that a tree metric of size $n$ can be embedded into Euclidean
$d'$-dimensional space with distortion in $O(n^{\frac{1}{d'-1}})$.
Thus for any $d'> 2d+1$, we cannot obtain an approximation factor
$o(n^{\frac{1}{2d}-\frac{1}{d'-1}})$ for $k$-forest in Euclidean
$d'$-dimensional space, since otherwise by solving the problem in
Euclidean $d'$-dimensional space, finding an Eulerian tour and
shortcutting it, and finally embedding it back into the star, we can
obtain a better approximation than $O(n^{\frac{1}{2d}})$, a
contradiction.
\end{proof}

Note as mentioned above that, despite extensive study, 
the current best approximation factor 
for notorious densest $k$-subgraph is
$O(n^{1/3-\epsilon})$~\cite{FKP01} and thus we do not expect
to have any PTAS for $k$-forest in $8$-dimensional Euclidean space.


Unlike the general cases of these problems,
 as far as PTASs for the case of Euclidean spaces are concerned,
 it seems \prob{$k$-forest} and \prob{prize-collecting Steiner forest} problems are essentially equivalent.
 Indeed in Lemma~\ref{lem:kfor>pcsf}, we prove that any PTAS for \prob{$k$-forest} results in a PTAS for \prob{prize-collecting Steiner
 forest},
 and we believe that any DP algorithm giving a PTAS for PCSF computes
 along its way the optimal solution to different $k$-forest instances.

Thus based on the evidences above, we do believe \prob{Euclidean
$k$-forest} and \prob{Euclidean prize-collecting Steiner forest}
have no PTASs in their general forms.


\section{Conclusion}
\label{sec:conclusion} Besides presenting a simpler and correct
analysis of the PTAS for the \prob{Euclidean Steiner forest
problem}, we showed how the approach can be generalized to solve
multiplicative prize-collecting problems.

Generalizing our results to planar graphs, especially obtaining a
PTAS for Steiner forest, 
has been a long-standing open problem in this field.
The question was settled very recently by Bateni, Hajiaghayi and Marx~\cite{BHM09:forest}.
While Borradaile, Klein and Kenyon-Mathieu~\cite{BKM07} gave a
PTAS for \prob{Steiner tree} on planar graphs, a main ingredient
of their algorithm is solving \prob{Steiner tree} on graphs of
bounded-treewidth. However in a sharp contrast, Gassner~\cite{Gas08}
showed recently that \prob{Steiner forest} is NP-hard even on graphs of
treewidth at most 3.
Bateni et al.~\cite{BHM09:forest} 
gives a PTAS for the problem on graphs of bounded treewidth,
and uses it to obtain a PTAS for planar and bounded-genus graphs.

Last but not least, obtaining any improvement over the approximation
factor 2.54 in~\cite{HJ06} for multiplicative prize-collecting
Steiner forest in general graphs seems very interesting.

%

{
\bibliographystyle{siam} 
\bibliography{main2}
}

\appendix

\section{Deferred proofs and further discussion}

\begin{proofof}{\proofname\ of Lemma~\ref{lem:tot-int}}
 The proof of Theorem~\ref{thm:twoprops} (although not reproduced here) does not increase $\comp(F\cap\eL)$.
 Hence, it suffices to prove the result for the forest $F$ specified at the beginning of Section~\ref{sec:prelim}.
 Observe that by (II),  $F\cap\eL$ consists merely of singleton points,
 because no Steiner point lies on a line $\ell\in\eL$.
 Further notice that the $\ell_1$-length of $F$ is at most $\sqrt 2\len(F)$.
 Let $F^x$ be the total absolute distance $F$ travels in the $x$ direction. 
 Since $x$-coordinate difference of any two consecutive \emph{break-points} of $F$ is a multiple of $2$ and the intersection with vertical lines of $\eL$ occurs at coordinates of the form $(2i, y)$, the total number of intersections with vertical lines is exactly $F^x/2$.
 We can similarly argue for the intersections with horizontal lines, and finally conclude that  $\comp(F\cap\eL)\leq\frac{\sqrt 2}{2}\len(F)$.
\end{proofof}

\begin{proofof}{\proofname\ of Lemma~\ref{lem:dp-sf}}
This is clearly true for the base cases of the DP since we go over all the possibilities.
Next, take any configuration $\chi=(\mathcal{K},\mathcal{P})$ corresponding to a non-leaf dissection square $R$,
and suppose there is a subsolution $F$ with respect to $R$ compatible with $\chi$,
 such that any subsolution $F'$ formed by restricting $F$ to a dissection square $R'$ which is a descendant of $R$ is compatible with some configuration $\chi'$ of $R'$.
Let $F_i$ be the subsolutions restricted to the subsquares $R_i$.
Each of them is thus compatible with an appropriate $\chi_i=(\mathcal{K}_i,\mathcal{P}_i)$.
By inductive hypothesis, the dynamic programming states $T_{R_i}[\chi_i]$ have been correctly computed.
Each connected component of $F$ not connected to $\partial R$ has to have all its terminal pairs satisfied.
This is taken care of by checking the partition $\mathcal{P}'$:
the terminals in components that do not advance in the dynamic table to $T_R[\chi]$ have their demands satisfied internally.
In addition, the locality property for $F$ ensures the configurations will be consistent,
and hence we perform an update of $T_R[\chi]$ from $T_{R_i}[\chi_i]$.
This finishes the completeness proof.

Verifying that the update rule is sound is trivial.
If the four configurations $\chi_i$ update $\chi$,
then there exists a subsolution $F$ formed by the union of the corresponding subsolutions $F_i$, that is compatible with $\chi$.
\end{proofof}

\subsection{The preliminary conditions for multiplicative prizes}\label{sec:prelim-mult}
In Section~\ref{sec:prelim}, we said that standard perturbation and scaling techniques allow us to assume
with a cost increase of at most $O(\eps\OPT)$
that the bounding box of the instance has side length at most $n^2\eps^{-1}\OPT$,
while restricting all vertices and Steiner points to points of the form $(2i+1,2j+1)$ for integers $i$ and $j$.
The claim is based on the following two premises:
\begin{enumerate}\setlength{\itemsep}{-.01in}
\item If $d$ is the maximum distance of a pair in $\DD$, then $\OPT \geq d$.
\item If $u$ and $v$
are farther than $n^2d$,
they cannot be connected in the optimum solution.
\end{enumerate}
Using this, the instance can be broken up into disjoint subinstances and then the perturbation can be carried out.
However, the first premise is false in the case of multiplicative prizes since not all the pairs need to be connected.
Next we show how similar conditions can be guaranteed in this case.

The value of $\OPT$ can be \emph{guessed} using binary search.
To begin the search, we can get crude bounds of $\omega/n \leq \OPT \leq \omega$,
using simple approximation algorithms for the general cases of PCSF and $k$-forest.\footnote{The best known approximation algorithms known for these problems are $2.54$ and $\min\{\sqrt k,\sqrt n\}$, respectively.}
Knowing $\OPT$, we build a graph $G'$ on the vertices:
there is an edge between $u$ and $v$ if and only if their distance is at most $\OPT$.
The diameter of each connected component is at most $n\OPT$.
We consider each of them separately, since two vertices in different components cannot be connected in the optimal solution.

The side length of the bounding box is at most $n\OPT$.
Scale the instance by $8\eps^{-1}$ and let $\OPT'=8\eps^{-1}\OPT$ denote the new optimal value.
Build a grid in the bounding by lines with equations $x=2i$ and $y=2j$ for integers $i,j$.
Move each vertex and Steiner point to the closest point of the form $(2i+1,2j+1)$.
Notice that there are at most $n$ Steiner points.
Assuming $\OPT>0$, the change in the solution value
due to the perturbation is at most $2n\cdot 4=8n\leq \eps\OPT'$.
Hence, we can assume that
\begin{itemize}\setlength{\itemsep}{-.01in}
\item the side length of the bounding box is at most $n\eps^{-1}\OPT'$, and
\item the vertices and Steiner points are at coordinates $(2i+1,2j+1)$ for integers $i,j$.
\end{itemize}

\subsection{$k$-MST as a special case of $S$-MPCSF}\label{sec:kmst}
Here we show that (even the symmetric) $S$-MPCSF is a generalization of the rooted $k$-MST problem (for which the best approximation guarantee is $2$).
Suppose we are given an instance $\mathcal{I}$ of the rooted $k$-MST problem.
It consists of a graph $G(V,E)$, edge lengths $c_e$, a root vertex $r$ and a number $k$.
Suppose $r$ is not to be counted among the $k$ vertices.
We build the new instance $\mathcal{I}'$ of the $S$-MPCSF problem as follows.
The graph $G'$ is the same as $G$.
The weights of all vertices are one, except for $r$ whose weight is $n^2$.
Then, the goal will be to find the cheapest forest that gathers a prize of at least $S=(n^2+k)^2=n^4+2n^2k+k^2$.

\begin{theorem}\label{thm:special-case}
The instance $\mathcal{I}$ of the rooted $k$-MST problem
is equivalent to
the instance $\mathcal{I}'$ of the $S$-MPCSF problem.
\end{theorem}
\begin{proof}
As we noted in Subsection~\ref{sec:contrib},
in case of polynomially bounded integer weights,
we can make sure the returned solution collects a prize of at least $S$ (without any approximation factor).
This can be achieved by picking $\eps' < 1/S$.

Obviously, any tree connecting $k$ vertices to the root is translated
to a forest that collects a prize of at least $S$.
Let each vertex not spanned by the tree be a singleton component in the forest.

Finally, we claim that any solution of value $S$ or higher translates to a solution of value at least $k$ for the original instance.
The resulting tree is just the component of the forest containing the root vertex.
Suppose for the sake of reaching a contradiction that the component spans $k'<k$ non-root vertices.
The total prize collected is at most
\begin{align*}
(n^2+k')^2 + (n-k'-1)^2  &< n^4+2n^2k'+k'^2+n^2 \\
                         &= S + 2n^2(k'-k) +k'^2 + n^2 - k^2\\
                         &< S+2n^2(k'-k+1) \\
                         &\leq S,
\end{align*}
yielding a contradiction, and proving the supposition is false.
\end{proof}



\subsection{PCSF vs. $k$-forest}
\begin{lemma}\label{lem:kfor>pcsf}
 An $\alpha$-approximation algorithm for the \prob{$k$-forest} problem
 gives an $\alpha(1+\eps)$-approximation algorithm for the \prob{prize-collecting Steiner forest} problem, for any constant $\eps > 0$.
\end{lemma}
\begin{proof}
 We show how to approximate a PCSF instance $\I$ by
 invoking several (polynomially many) instances $\I'$ of the $k$-forest problem.
 Obtain an estimate $\omega$ for  $\I$,
 such that $\frac{\omega}{3} \leq \OPT \leq \omega$ using a general-case $3$-approximation algorithm.
 Let $\pi_i$ be the penalty of the pair $i$ in $\I$.
 Without loss of generality, we can assume that $\pi_i \leq 2\omega$ for any pair $i$.
 Let $\theta = \eps\omega/3n$.
 Place $p_i=\lfloor \frac{\pi_i}{\theta} \rfloor$ copies  of  the pair $i$ in $\I'$. 
 Find an $\alpha$-approximate solution to the resulting $k$-forest instance
 for every value of $0\leq k\leq n'$, where $n'$ is the number of pairs in $\I'$.
 Compute the PCSF value for each of these solutions and report the best one.

 We show that at least one of these candidate solutions is good.
 Let $\OPT_f$ and $\OPT_\pi$ be the length of the forest and the paid penalty
 of the optimal solution, respectively.
 Suppose $\OPT$ connects a subset of terminal pairs $Q$.
 Then, $\OPT_\pi = \sum_{i\not\in Q}\pi_i$.
 Focus on the candidate solution with $k = \sum_{i\in Q}\lfloor\pi_i/\theta\rfloor$.
 The length of the corresponding $k$-forest instance is
 at most $\OPT_f$, because a possible solution is that of connecting the copies of  $Q$.
 To compute the PCSF value, we add the penalty of pairs in $Q$ that are not connected using this tree.
 We can assume either all or no copies of each pair is connected.
 The number of pairs not connected is at most $n' - k$,
 and  their penalties sum to no more than
\begin{align*}
      \sum_{i\text{ not connected}} \pi_i
&\leq  \sum_{i\text{ not connected}} (p_i+1)\theta \\
&\leq  \sum_{i\text{ not connected}} p_i\theta + n\theta \\
&\leq  (n'-k)\theta + n\theta \\
&\leq  \OPT_\pi + n\theta \\
&=     \OPT_\pi + \eps\omega/3 \\
&\leq  \OPT_\pi + \eps\OPT.
\end{align*}
Thus, the PCSF value of the best candidate solution is at most $\alpha\OPT_f+\OPT_\pi+\eps\OPT\leq\alpha(1+\eps)\OPT$.
 It remains to show the instances $\I'$ have polynomial size.
 Since $\pi \leq 2\omega$, each pair $i$ will have $p_i \leq 6n\eps^{-1}$ copies.
 Hence, $\I'$ has polynomial size and we can use the approximation algorithm for the $k$-forest.
\end{proof}

\end{document}